\documentclass[12pt]{article}
\usepackage{amsmath,amssymb,amsthm}
\usepackage{times}
\usepackage{graphicx}
\usepackage{color}
\usepackage{multirow}
\usepackage{sidecap}
\usepackage{hyperref}
\usepackage{enumitem}
\usepackage{url}
\usepackage[authoryear]{natbib}
\usepackage{rotating}
\usepackage{bbm}
\usepackage{latexsym}

\newcommand{\captionfonts}{\normalsize}

\makeatletter  
\long\def\@makecaption#1#2{%
  \vskip\abovecaptionskip
  \sbox\@tempboxa{{\captionfonts #1: #2}}%
  \ifdim \wd\@tempboxa >\hsize
    {\captionfonts #1: #2\par}
  \else
    \hbox to\hsize{\hfil\box\@tempboxa\hfil}%
  \fi
  \vskip\belowcaptionskip}
\makeatother   

\newtheorem*{rep@theorem}{\rep@title}
\newcommand{\newreptheorem}[2]{%
\newenvironment{rep#1}[1]{%
 \def\rep@title{#2 \ref{##1}}%
 \begin{rep@theorem}}%
 {\end{rep@theorem}}}
\makeatother

\newtheorem{Th}{Theorem} 
\newreptheorem{Th}{Theorem}
\newtheorem{lem}{Lemma}
\newreptheorem{lem}{Lemma}
\newtheorem{cor}{Corollary} 

\theoremstyle{definition}
\newtheorem{defn}{Definition}

\theoremstyle{remark}
\newtheorem*{rmk}{Remark}

\begin{document}

\begin{center}
{\Large \bf A Hebbian/Anti-Hebbian Neural Network for Linear Subspace Learning: A Derivation from Multidimensional Scaling of Streaming Data}

{Cengiz Pehlevan$^{\displaystyle 1,2}$, Tao Hu$^{\displaystyle 3}$, Dmitri B. Chklovskii$^{\displaystyle 2}$}\\
{$^{\displaystyle 1}$Janelia Research Campus, Howard Hughes Medical Institute, Ashburn, VA 20147}\\
{$^{\displaystyle 2}$Simons Center for Data Analysis, 160 Fifth Ave, New York, NY 10010 }\\
{$^{\displaystyle 3}$Texas A\&M University, MS 3128 TAMUS, College Station, TX 77843}\\
\end{center}

\vspace{0mm}

\thispagestyle{empty}
\markboth{}{NC instructions}
%
%
\begin{center} {\bf Abstract} \end{center}

Neural network models of early sensory processing typically reduce the dimensionality of streaming input data. Such networks learn the principal subspace, in the sense of principal component analysis (PCA), by adjusting synaptic weights according to activity-dependent learning rules. When derived from a principled cost function these rules are nonlocal and hence biologically implausible. At the same time, biologically plausible local rules have been postulated rather than derived from a principled cost function. Here, to bridge this gap, we derive a biologically plausible network for subspace learning on streaming data by minimizing a principled cost function. In a departure from previous work, where cost was quantified by the representation, or reconstruction, error, we adopt a  multidimensional scaling  (MDS) cost function for streaming data. The resulting algorithm relies only on biologically plausible Hebbian and anti-Hebbian local learning rules. In a stochastic setting, synaptic weights converge to a stationary state which projects the input data onto the principal subspace. If the data are generated by a nonstationary distribution, the network can track the principal subspace. Thus, our result makes a step towards an algorithmic theory of neural computation.

\section{Introduction}

Early sensory processing reduces the dimensionality of streamed inputs \citep{hyvarinen2009} as evidenced by a high ratio of input to output nerve fiber counts \citep{shepherd2003}. For example, in the human retina, information gathered by $\approx$125 million photoreceptors is conveyed to the Lateral Geniculate Nucleus through $\approx$1 million ganglion cells \citep{hubel1995}. By learning a lower-dimensional subspace and projecting the streamed data onto that subspace the nervous system de-noises and compresses the data simplifying further processing. Therefore, a biologically plausible implementation of dimensionality reduction may offer a model of early sensory processing. 

For a single neuron, a biologically plausible implementation of dimensionality reduction in the streaming, or online, setting has been proposed in the seminal work of \citep{oja1982simplified}, Figure \ref{Fig1}A. At each time point, $t$, an input vector, ${\bf x}_t$, is presented to the neuron, and, in response, it computes a scalar output, $y_t={\bf w} {\bf x}_t$, were {\bf w} is a row-vector of input synaptic weights. Furthermore, synaptic weights {\bf w} are updated according to a version of Hebbian learning called Oja's rule:
\begin{align}\label{singleOja}
{\bf w} \leftarrow {\bf w} + \eta y_t ({\bf x}_t^\top-{\bf w} y_t),
\end{align}
where $\eta$ is a learning rate and $^\top$ designates a transpose. Then, the neuron's synaptic weight vector converges to the principal eigenvector of the covariance matrix of the streamed data \citep{oja1982simplified}. Importantly, Oja's learning rule is local meaning that synaptic weight updates depend on the activities of only pre- and postsynaptic neurons accessible to each synapse and, therefore, biologically plausible. 

Oja's rule can be derived by an approximate gradient descent of the mean squared representation error \citep{cichocki2002adaptive,yang1995projection}, a so-called synthesis view of principal component analysis (PCA) \citep{pearson1901,preisendorfer1988principal}:
\begin{align}\label{singlePCA}
\min_{{\bf w}}\sum_t  \left\Vert {\bf x}_t-{\bf w^\top}{\bf w}{\bf x}_t\right\Vert_2^2.
\end{align}

Computing principal components beyond the first requires more than one output neuron and motivated numerous neural networks. Some well-known examples are the Generalized Hebbian Algorithm (GHA) \citep{sanger1989optimal}, F\"oldiak's network \citep{foldiak1989adaptive}, the subspace network \citep{karhunen1982new}, Rubner's network \citep{rubner1989self,rubner1990development}, Leen's minimal coupling and full coupling networks \citep{leen1990,leen1991} and the APEX network \citep{kung1990neural,kung1994adaptive}. We refer to \citep{BeckerPlumbley96,diamantaras1996principal,diamantaras2002neural} for a detailed review of these and further developments. 

However, none of the previous contributions was able to derive a multineuronal single-layer network with local learning rules by minimizing a principled cost function, in a way that Oja's rule \eqref{singleOja} was derived for a single neuron. The GHA and the subspace rules rely on nonlocal learning rules: feedforward synaptic updates depend on other neurons' synaptic weights and activities. Leen's minimal network is also nonlocal:  feedforward synaptic updates of a neuron depend on its lateral synaptic weights. While F\"oldiak's, Rubner's and Leen's full coupling networks use local Hebbian and anti-Hebbian rules, they were postulated rather than derived from a principled cost function. APEX network, perhaps, comes closest to our criterion: the rule for each neuron can be related separately to a cost function which includes contributions from other neurons. But no cost function describes all the neurons combined.

At the same time, numerous dimensionality reduction algorithms have been developed for data analysis needs disregarding the biological plausibility requirement. Perhaps the most common approach is again PCA, which was originally developed for batch processing \citep{pearson1901} but later adapted to streaming data \citep{yang1995projection,crammer2006online,arora2012stochastic,goes2014robust}. For a more detailed collection of references, see e.g. \citep{balzano2012}. These algorithms typically minimize the representation error cost function:
\begin{align}\label{PCA}
\min_{\bf F} \left\Vert {\bf X}-{\bf F^\top}{\bf F}{\bf X}\right\Vert_F^2,
\end{align}
where ${\bf X}$ is a data matrix and ${\bf F}$ is a wide matrix (for detailed notation, see below). The minimum of \eqref{PCA} is when rows of ${\bf F}$ are orthonormal and span the $m$-dimensional principal subspace, and therefore ${\bf F}^\top{\bf F}$ is the projection matrix to the subspace \citep{yang1995projection}\footnote{Recall that, in general, the projection matrix to the row space of a matrix {\bf P} is given by $ {\bf P}^\top\left({\bf P}{\bf P}^\top\right)^{-1}{\bf P}$,  provided ${\bf P}{\bf P}^\top$ is full rank \citep{plumbley1995}. If the rows of $ {\bf P}$ are orthonormal this reduces to ${\bf P}^\top{\bf P}$.}.

 A gradient descent minimization of such cost function can be approximately implemented by the subspace network \citep{yang1995projection}, which, as pointed out above, requires nonlocal learning rules. While this algorithm can be implemented in a neural network using local learning rules, it requires a second layer of neurons  \citep{oja1992principal}, making it less appealing.

In this paper, we derive a single-layer network with local Hebbian and anti-Hebbian learning rules, similar in architecture to F\"oldiak's \citep{foldiak1989adaptive} (see Figure \ref{Fig1}B),  from a principled cost function and demonstrate that it recovers a principal subspace from streaming data. The novelty of our approach is that, rather than starting with the representation error cost function traditionally used for dimensionality reduction, such as PCA, we use the cost function of classical multidimensional scaling (CMDS), a member of the family of multidimensional scaling (MDS) methods \citep{cox2000multidimensional,mardia1980multivariate}. Whereas the connection between CMDS and PCA has been pointed out previously \citep{williams01ona,cox2000multidimensional,mardia1980multivariate}, CMDS is typically performed in the batch setting. Instead, we developed a neural network implementation of CMDS for streaming data.  

The rest of the paper is organized as follows. In Section 2, by minimizing the CMDS cost function we derive two online algorithms implementable by a single-layer network, with synchronous and asynchronous synaptic weight updates. In Section 3, we demonstrate analytically that synaptic weights define a principal subspace whose dimension $m$ is determined by the number of output neurons and that the stability of the solution requires that this subspace corresponds to top $m$ principal components. In Section 4, we show numerically that our algorithm recovers the principal subspace of a synthetic dataset, and does it faster than the existing algorithms. Finally, in Section 5, we consider the case when data are generated by a nonstationary distribution and present a generalization of our algorithm to  principal subspace tracking.

\begin{figure}
\centering
\includegraphics{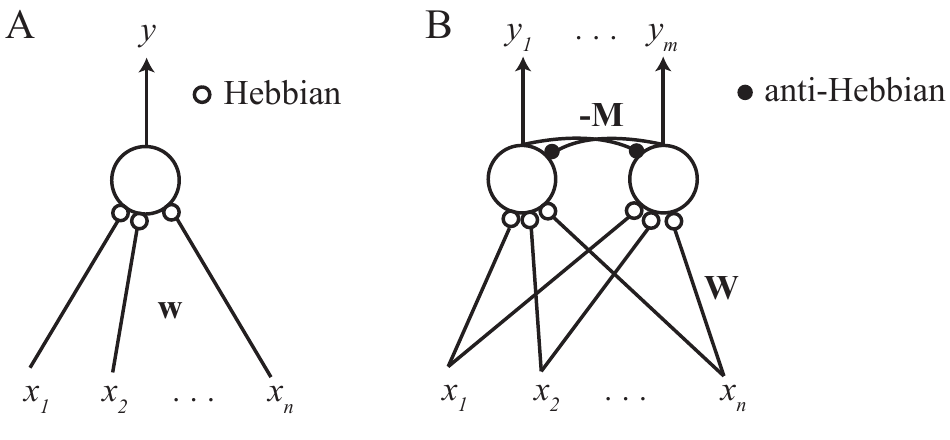}
  \caption{ An Oja  neuron and our neural network. {\bf A}. A single Oja neuron computes the principal component, $y$, of the input data, ${\bf x}$, if its synaptic weights follow Hebbian updates. {\bf B}. A multineuron network computes the principal subspace of the input if the feedforward connection weight updates follow a Hebbian and the lateral connection weight updates follow an anti-Hebbian rule.\label{Fig1}}
\end{figure}

\section{Derivation of online algorithms from the CMDS cost function}

CMDS represents high-dimensional input data in a lower-dimensional output space while preserving pairwise similarities between samples\footnote{Whereas MDS in general starts with dissimilarities between samples that may not live in Euclidean geometry, in CMDS data are assumed to have a Euclidean representation.} \citep{young1938discussion,torgerson1952multidimensional}.
Let $T$ centered input data samples in $\mathbb{R}^n$ be represented by column-vectors ${\bf x}_{t=1,\ldots,T}$ concatenated into an $n\times T$ matrix ${\bf X}=[{\bf x}_1,\ldots,{\bf x}_T]$. The corresponding output representations in $\mathbb{R}^m$, $m\leq n$, are column-vectors, ${\bf y}_{t=1,\ldots,T}$, concatenated into an $m\times T$ dimensional matrix ${\bf Y}=[{\bf y}_1,\ldots,{\bf y}_T]$. Similarities between vectors in Euclidean spaces are captured by their inner products. For the input (output) data, such inner products are assembled into a ${T\times T}$ Gram matrix\footnote{\label{dissimilarity}When input data are pairwise Euclidean distances, assembled into a matrix ${\bf Q}$, the Gram matrix, ${\bf X}^\top{\bf X}$, can be constructed from ${\bf Q}$ by ${\bf H} {\bf Z} {\bf H}$, where $Z_{ij} = -1/2Q^2_{ij}$, ${\bf H}={\bf I}_n-1/n{\bf1}{\bf1}^\top$ is the centering matrix, and ${\bf I}_n$ is the $n$ dimensional identity matrix \citep{cox2000multidimensional,mardia1980multivariate}.} ${\bf X}^\top{\bf X}$ (${\bf Y}^\top{\bf Y}$). For a given ${\bf X}$, CMDS finds ${\bf Y}$ by minimizing the so-called ``strain" cost function \citep{carroll1972idioscal} :
\begin{align}\label{strain}
 \min_{{\bf Y}} \left\Vert {\bf X}^\top{\bf X}-{\bf Y}^\top{\bf Y}\right\Vert_F^2.
\end{align}

For discovering a low-dimensional subspace, the CMDS cost function \eqref{strain} is a viable alternative to the representation error cost function \eqref{PCA} because its solution is related to PCA \citep{williams01ona,cox2000multidimensional,mardia1980multivariate}. Specifically, {\bf Y} is the linear projection of {\bf X} onto the (principal sub-)space  spanned by $m$ principal eigenvectors of the sample covariance matrix ${\bf C}_T= \frac 1T \sum_{t=1}^T {\bf x}_t {\bf x}_t^\top={\bf X}{\bf X}^\top$. The CMDS cost function defines a subspace rather than individual eigenvectors because left orthogonal rotations of an optimal {\bf Y} stay in the subspace and are also optimal, as is evident from the symmetry of the cost function. 

In order to reduce the dimensionality of streaming data, we minimize the CMDS cost function \eqref{strain} in the stochastic online setting.  At time $T$, a data sample, ${\bf x}_T$, drawn independently from a zero-mean distribution is presented to the algorithm which computes a corresponding output, ${\bf y}_T$ prior to the presentation of the next data sample. Whereas in the batch setting, each data sample affects all outputs, in the online setting, past outputs cannot be altered. Thus, at time $T$ the algorithm minimizes the cost depending on all inputs and ouputs up to time $T$ with respect to ${\bf y}_T$ while keeping all the previous outputs fixed:
\begin{align}\label{Frob}
{{\bf y}_T} &=  \mathop {\arg \min }\limits_{{\bf y}_T} \left\Vert {\bf X}^\top{\bf X}-{\bf Y}^\top{\bf Y}\right\Vert_F^2= \mathop {\arg \min }\limits_{{\bf y}_T} \sum_{t=1}^T\sum_{t'=1}^{T}\left({\bf x}_t^\top{\bf x}_{t'}-{\bf y}_t^\top{\bf y}_{t'}\right)^2, 
\end{align}
where the last equality follows from the definition of the Frobenius norm. By keeping only the terms that depend on current output ${\bf y}_T$ we get:
\begin{align}\label{onlineStrain_old} 
{{\bf y}_T}  &=\mathop {\arg \min }\limits_{{\bf y}_T} \left[ { - 4{{\bf x}^\top_T}\left( {\sum\limits_{t = 1}^{T - 1} {{\bf x}_t{{\bf y}^\top_t}} } \right){\bf y}_T + 2{{\bf y}^\top_T}\left( {\sum\limits_{t = 1}^{T - 1} {{\bf y}_t{{\bf y}^\top_t}} } \right){\bf y}_T  - 2{{\left\| {{\bf x}_T} \right\|}^2}{{\left\| {{\bf y}_T} \right\|}^2} + {{\left\| {{\bf y}_T} \right\|}^4}} \right].
\end{align}

In the large-$T$ limit, expression \eqref{onlineStrain_old} simplifies further because the first two terms grow linearly with $T$, and therefore dominate over the last two. After dropping the last two terms we arrive at: 
 \begin{align}\label{onlineStrain}
{{\bf y}_T} =\mathop {\arg \min }\limits_{{\bf y}_T} \left[ { - 4{{\bf x}^\top_T}\left( {\sum\limits_{t = 1}^{T - 1} {{\bf x}_t{{\bf y}^\top_t}} } \right){\bf y}_T + 2{{\bf y}^\top_T}\left( {\sum\limits_{t = 1}^{T - 1} {{\bf y}_t{{\bf y}^\top_t}} } \right){\bf y}_T  } \right].
\end{align}

We term the cost in expression \eqref{onlineStrain} the ``online CMDS cost". Because the online CMDS cost is a positive semi-definite quadratic form in ${\bf y}_T$ for sufficiently large $T$, this optimization problem is convex.  While it admits a closed-form analytical solution via matrix inversion, we are interested in biologically plausible algorithms. Next, we consider two algorithms that can be mapped onto single-layer neural networks with local learning rules: coordinate descent leading to asynchronous updates and Jacobi iteration leading to synchronous updates. 

\subsection{A neural network with asynchronous updates}

The online CMDS cost function \eqref{onlineStrain} can be minimized by  coordinate descent which at every step finds the optimal value of one component of ${\bf y}_T$ while keeping the rest fixed. The components can be cycled through in any order until the iteration converges to a fixed point. Such iteration is guaranteed to converge under very mild assumptions:  diagonals of $\sum\limits_{t = 1}^{T - 1} {{\bf y}_t{{\bf y}^\top_t}}$ has to be positive \citep{luo1991on}, meaning that each output coordinate has produced at least one non-zero output before current time step $T$. This condition is almost always satisfied in practice. 

The cost to be minimized at each coordinate descent step with respect to $i^{\rm th}$ channel's activity is:
 \begin{align*}
{y_{T,i}} =\mathop {\arg \min }\limits_{{y_{T,i}} } \left[ { - 4{{\bf x}^\top_T}\left( {\sum\limits_{t = 1}^{T - 1} {{\bf x}_t{{\bf y}^\top_t}} } \right){\bf y}_T + 2{{\bf y}^\top_T}\left( {\sum\limits_{t = 1}^{T - 1} {{\bf y}_t{{\bf y}^\top_t}} } \right){\bf y}_T  } \right].
\end{align*}
Keeping only those terms that depend on ${y_{T,i}}$ yields:
 \begin{align*}
{y_{T,i}} &=\mathop {\arg \min }\limits_{{y_{T,i}} } \left[ - 4\sum_k{x_{T,k}} \left({\sum\limits_{t = 1}^{T - 1} {x_{t,k}y_{t,i}} }\right) y_{T,i} \nonumber \right. \\
&\qquad\qquad\qquad\qquad\left.+ 4\sum_{j\neq i}{y_{T,j}}\left( {\sum\limits_{t = 1}^{T - 1} {y_{t,j}{y_{t,i}}} } \right)y_{T,i} + 2\left( {\sum\limits_{t = 1}^{T - 1} y_{t,i}^2}  \right)y_{T,i}^2   \right].
\end{align*}
By taking a derivative with respect to  $y_{T,i}$ and setting it to zero we arrive at the following closed-form solution:
\begin{align}
\label{solution1}
{y_{T,i}} =  \frac{{\sum\limits_k {\left( {\sum\limits_{t = 1}^{T - 1} {y_{t,i}^{}x_{t,k}^{}} } \right)x_{T,k}^{}} }}{{\sum\limits_{t = 1}^{T - 1} {y_{t,i}^2} }} - \frac{{\sum\limits_{j \ne i} {\left( {\sum\limits_{t = 1}^{T - 1} {y_{t,i}^{}y_{t,j}^{}} } \right)y_{T,j}^{}} }}{{\sum\limits_{t = 1}^{T - 1} {y_{t,i}^2} }}.
\end{align}

To implement this algorithm in a neural network we denote normalized input-output and output-output covariances, 
\begin{align}\label{WM}
W_{T,ik}^{} = \frac{{\sum\limits_{t = 1}^{T - 1} {y_{t,i}^{}x_{t,k}^{}} }}{{\sum\limits_{t = 1}^{T - 1} {y_{t,i}^2} }},\qquad  M_{T,i,j \ne i}^{} = \frac{{\sum\limits_{t = 1}^{T - 1} {y_{t,i}^{}y_{t,j}^{}} }}{{\sum\limits_{t = 1}^{T - 1} {y_{t,i}^2} }},\qquad M_{T,ii}^{} = 0,
\end{align}
allowing us to rewrite the solution \eqref{solution1} in a form suggestive of a linear neural network: 
\begin{align}\label{dynamics}
{y_{T,i}} \leftarrow \sum_{j=1}^n W_{T,ij} x_{T,j} - \sum_{j=1}^m M_{T,ij} y_{T,j}, 
\end{align}
where ${\bf W}_T$ and ${\bf M}_T$  represent the synaptic weights of feedforward and lateral connections respectively, Figure 1B.

Finally, to formulate a fully online algorithm we rewrite \eqref{WM} in a recursive form. This requires introducing a scalar variable ${D}_{T,i}$ representing cumulative activity of a neuron $i$ up to time $T-1$, 
\begin{align}\label{Y}
{D}_{T,i} =\sum\limits_{t = 1}^{T - 1} {y_{t,i}^2},
\end{align}
Then, at each time point, $T$, after the output ${\bf y}_T$ is computed by the network, the following updates are performed:
\begin{align}\label{hah}
D_{T+1,i} &\leftarrow D_{T,i}+ y_{T,i}^2\nonumber\\
{W_{T+1,ij}} &\leftarrow {W_{T ,ij}} + y_{T,i}\left( {x_{T,j} - {W_{T,ij}}y_{T,i}} \right)/D_{T+1,i} \nonumber\\
{M_{T+1,i,j \ne i}} &\leftarrow {M_{T,ij}} + y_{T,i}\left( {y_{T,j}^{} - {M_{T,ij}}y_{T,i}^{}} \right)/D_{T+1,i}.
\end{align}

Equations \eqref{dynamics} and \eqref{hah} define a neural network algorithm that minimizes the online CMDS cost function \eqref{onlineStrain} for streaming data by alternating between two phases: neural activity dynamics and synaptic updates. After a data sample is presented at time $T$, in the neuronal activity phase, neuron activities are updated one-by-one, i.e. asynchronously, \eqref{dynamics} until the dynamics converges to a fixed point defined by the following equation:
\begin{align}\label{fixed}
{\bf y}_T = {\bf W}_T{\bf x}_T-{\bf M}_T{\bf y}_T \qquad \implies \qquad {\bf y}_T = ({\bf I}_m+{\bf M}_T)^{-1}{\bf W}_T{\bf x}_T,
\end{align}
where ${\bf I}_m$ is the $m$-dimensional identity matrix.

In the second phase of the algorithm, synaptic weights are updated, according to a local Hebbian rule \eqref{hah} for feedforward connections, and according to a local anti-Hebbian rule (due to the $(-)$ sign in equation \eqref{dynamics}) for lateral connections. Interestingly, these updates have the same form as the single-neuron Oja's rule \eqref{singleOja} \citep{oja1982simplified}, except that the learning rate is not a free parameter but is determined by the cumulative neuronal activity $1/D_{T+1,i}$\footnote{\label{StepSize}The single neuron Oja's rule derived from the minimization of a least squares optimization cost function ends up with the identical learning rate \citep{diamantaras2002neural,hu2013neuron}. Motivated by this fact, such learning rate has been argued to be optimal for the APEX network \citep{diamantaras1996principal,diamantaras2002neural} and used by others \citep{yang1995projection}.}. To the best of our knowledge such single-neuron rule \citep{hu2013neuron} has not been derived in the multineuron case. An alternative derivation of this algorithm is presented in Appendix \ref{altAsynch}

Unlike the representation error cost function \eqref{PCA}, the CMDS cost function \eqref{strain} is formulated only in terms of input and output activity. Yet, the minimization with respect to ${\bf Y}$ recovers feedforward and lateral synaptic weights. 

\subsection{A neural network with synchronous updates}

Here, we present an alternative way to derive a neural network algorithm from the large-$T$ limit of the online CMDS cost function \eqref{onlineStrain} . By taking a derivative with respect to  ${\bf y}_T$ and setting it to zero we arrive at the following linear matrix equation:
\begin{align}\label{anMin}
\left( {\sum\limits_{t = 1}^{T - 1} {{\bf y}_t{{\bf y}^\top_t}} } \right){\bf y}_T  = \left( {\sum\limits_{t = 1}^{T - 1} {{\bf y}_t{{\bf x}^\top_t}} } \right){{\bf x}_T},
 \end{align} 
We solve this system of equations using Jacobi iteration \citep{strang2009}, by first splitting the output covariance matrix that appears on the left side of \eqref{anMin} into its diagonal component ${\bf D}_T$ and the remainder ${\bf R}_T$:
\begin{align*}
\left( {\sum\limits_{t = 1}^{T - 1} {{\bf y}_t{{\bf y}^\top_t}} } \right) = {\bf D}_T + {\bf R}_T,
\end{align*}
where the $i$ th diagonal element of ${\bf D}_T$, $D_{T,i} = \sum_{t=1}^{T-1}y_{t,i}^2$, as defined in \eqref{Y}. Then, \eqref{anMin} is equivalent to:
\begin{align*}
{\bf y}_T  ={\bf D}_T^{-1} \left( {\sum\limits_{t = 1}^{T - 1} {{\bf y}_t{{\bf x}^\top_t}} } \right){{\bf x}_T} -{\bf D}_T^{-1}{\bf R}_T {\bf y}_T  .
\end{align*}

Interestingly, the matrices obtained on the right side are algebraically equivalent to the feedforward and lateral synaptic weight matrices defined in \eqref{WM}: 
\begin{align}
{\bf W}_T={\bf D}_T^{-1} \left( {\sum\limits_{t = 1}^{T - 1} {{\bf y}_t{{\bf x}^\top_t}} } \right) \qquad {\rm and } \qquad {\bf M}_T={\bf D}_T^{-1}{\bf R}_T.
\end{align} 
Hence, the Jacobi iteration for solving  \eqref{anMin}
\begin{align}\label{Jit}
{\bf y}_T \leftarrow {\bf W}_T{\bf x}_T-{\bf M}_T{\bf y}_T.
\end{align}
converges to the same fixed point as the coordinate descent, \eqref{fixed}.

Iteration \eqref{Jit} is naturally implemented by the same single-layer linear neural network as for the asynchronous update, Figure 1B. For each stimulus presentation the network goes through two  phases. In the first phase, iteration \eqref{Jit} is repeated until convergence. Unlike the coordinate descent algorithm which updated activity of neurons one after another, here, activities of all neurons are updated synchronously. In the second phase, synaptic weight matrices are updated according to the same rules as in the asynchronous update algorithm \eqref{hah}.

Unlike the asynchronous update \eqref{dynamics}, for which convergence is almost always guaranteed \citep{luo1991on}, convergence of iteration \eqref{Jit} is guaranteed only when the  spectral radius of ${\bf M}$ is less than 1 \citep{strang2009}. Whereas we cannot prove that this condition is always met, in practice, the synchronous algorithm works well. While in the rest of the paper, we consider only the asynchronous updates algorithm, our results hold for the synchronous updates algorithm provided it converges.

\section{Stationary synaptic weights define a principal subspace}\label{secAnalytical}

What is the nature of the lower dimensional representation found by our algorithm? In CMDS, outputs $y_{T,i}$ are the Euclidean coordinates in the principal subspace of the input vector ${\bf x}_T$ \citep{cox2000multidimensional,mardia1980multivariate}. While our algorithm uses the same cost function as CMDS, the minimization is performed in the streaming, or online, setting. Therefore, we cannot take for granted that our algorithm will find the principal subspace of the input.  In this section, we provide analytical evidence, by a stability analysis in a stochastic setting, that our algorithm extracts the principal subspace of the input data and projects onto that subspace. We start by previewing our results and method.

Our algorithm performs a linear dimensinality reduction since the transformation between the input and the output is linear. This can be seen from the neural activity fixed point \eqref{fixed}, which we rewrite as
\begin{align}\label{yfx}
{\bf y}_T = {\bf F}_T{\bf x}_T,
\end{align}
where  ${\bf F}_T$ is a matrix defined in terms of the synaptic weight matrices ${\bf W}_T$ and ${\bf M}_T$: 
\begin{align}\label{QF}
{\bf F}_T:= \left({\bf I}_m + {\bf M}_T\right)^{-1} {\bf W}_T.
\end{align}
Relation \eqref{yfx} shows that the linear filter of a neuron, which we term a ``neural filter", is the corresponding row of ${\bf F}_T$. The space that neural filters span, the rowspace of ${\bf F}_T$, is termed a ``filter space".   

First, we prove that in the stationary state of our algorithm, neural filters are indeed orthonormal vectors (section \ref{sTh1}, Theorem \ref{Th1}). Second, we demonstrate that the orthonormal filters form a basis of a space spanned by some $m$ eigenvectors of the covariance of the inputs ${\bf C}$ (section \ref{sTh2}, Theorem \ref{Th2}). Third, by analyzing linear perturbations around the stationary state, we find that stability requires these $m$ eigenvectors to be the principal eigenvectors and, therefore, the filter space to be the principal subspace (section \ref{sTh3}, Theorem \ref{Th3}). 

These results show that even though our algorithm was derived starting from the CMDS cost function \eqref{strain}, ${\bf F}_T$ converges to the optimal solution of the representation error cost function \eqref{PCA}. This correspondence suggests that ${\bf F}_T^\top{\bf F}_T$ is the algorithm's current estimate of the projection matrix to the principal subspace. Further, in \eqref{PCA}, columns of ${\bf F}^\top$ are interpreted as data features. Then, columns of ${\bf F}_T^\top$, or neural filters, are the algorithm's estimate of such features.

 Rigorous stability analyses of PCA neural networks  \citep{oja1982simplified,oja1985stochastic,sanger1989optimal,oja1992principal,hornik1992convergence,plumbley1995} typically use the ODE method \citep{kushner1978}: Using a theorem of stochastic approximation theory \citep{kushner1978}, the convergence properties of the algorithm are determined using a corresponding deterministic differential equation\footnote{Application of stochastic approximation theory to PCA neural networks depends on a set of mathematical assumptions. See \citep{zufiria2002} for a critique of the validity of these assumptions and an alternative approach to stability analysis.}. 
 
 Unfortunately the ODE method cannot be used for our network. While the method requires learning rates that depend only on time, in our network learning rates ($1/D_{T+1,i}$) are activity dependent. Therefore we take a different approach. We directly work with the discrete-time system, assume convergence to a ``stationary state", to be defined below, and study the stability of the stationary state.
 
\subsection{Preliminaries}

We adopt a stochastic setting where the input to the network at each time point, ${\bf x}_t$, is an $n$-dimensional i.i.d. random vector with zero mean, $\left<{\bf x}_t\right>=0$, where brackets denote an average over the input distribution, and covariance ${\bf C} = \left<{\bf x}_t{\bf x}_t^\top\right>$.

Our analysis is performed for the ``stationary state" of synaptic weight updates, i.e. when averaged over the distribution of input values, the updates on ${\bf W}$ and ${\bf M}$ average to zero. This is the point of convergence of our algorithm. For the rest of the section, we drop the time index $T$ to denote stationary state variables.

The remaining dynamical variables, learning rates  $1/D_{T+1,i}$, keep decreasing at each time step due to neural activity. We assume that  the algorithm has run for a sufficiently long time such that the change in learning rate is small and it can be treated as a constant for a single update. Moreover, we assume that the algorithm converges to a stationary point sufficiently fast such that the following approximation is valid at large $T$:
\begin{align*}
\frac1{D_{T+1,i}} = \frac{1}{\sum_{t=1}^{T}y^2_{t,i}} \approx \frac{1}{T\left<y^2_{i}\right>},
\end{align*}
where ${\bf y}$ is calculated with stationary state weight matrices. 

We collect these assumptions into a definition.
\begin{defn}[Stationary State]\label{def1} 
In the stationary state, 
\begin{align*}
\left<\Delta W_{ij}\right> = \left<\Delta M_{ij}\right> = 0,
\end{align*}
and 
\begin{align*}
\frac1{D_{i}}  = \frac{1}{T\left<y^2_{i}\right>},
\end{align*}
with $T$ large.
\end{defn}

The stationary state assumption leads us to define various relations between synaptic weight matrices, summarized in the following corollary:
\begin{cor} In the stationary state,
\begin{align}\label{main1}
\left<y_ix_j\right>=\left<y_i^2\right>W_{ij},
\end{align}
and
\begin{align}\label{main2}
\left<y_iy_j\right>=\left<y_i^2\right> (M_{ij}+\delta_{ij}),
\end{align}
where $\delta_{ij}$ is the Kronecker-delta.
\end{cor}
\begin{proof}
Stationarity assumption when applied to the update rule on ${\bf W}$ \eqref{hah} leads immediately to \eqref{main1}. Stationarity assumption applied to the update rule on ${\bf M}$ \eqref{hah} gives:
\begin{align*}
\left<y_iy_j\right>=\left<y_i^2\right> M_{ij}, \qquad i\neq j. 
\end{align*}
The last equality does not hold for $i=j$ since diagonal elements of ${\bf M}$ are zero. To cover the case $i=j$, we add an identity matrix to ${\bf M}$, and hence one recovers \eqref{main2}.  
\end{proof}
\begin{rmk}Note that \eqref{main2} implies $\left<y_i^2\right> M_{ij} = \left<y_j^2\right> M_{ji}$, i.e. that lateral connection weights are not symmetrical.
\end{rmk}

\subsection{Orthonormality of neural filters}\label{sTh1}

Here we prove the orthonormality of neural filters in the stationary state. First, we need the following lemma:
\begin{lem}\label{lem1} In the stationary state, the following equality holds:
\begin{align}\label{first}
 {\bf I}_m +  {\bf M} = {\bf W} {\bf F}^\top.
\end{align}
\end{lem} 
\begin{proof}  By \eqref{main2}, $\left<y_i^2\right>\left(M_{ik}+\delta_{ik}\right) = \left<y_iy_k\right>$. Using ${\bf y} = {\bf F} {\bf x}$, we substitute for $y_k$ on the right hand side: $\left<y_i^2\right>\left(M_{ik}+\delta_{ik}\right)=  \sum_{j}F_{kj}\left<y_ix_j\right>$. Next, the stationarity condition \eqref{main1} yields: $\left<y_i^2\right>\left(M_{ik}+\delta_{ik}\right)=  \left<y_i^2\right>\sum_{j}F_{kj}W_{ij}$. Canceling $ \left<y_i^2\right>$ on both sides proves the Lemma.
\end{proof}
Now we can prove our theorem.
\begin{Th}\label{Th1}
In the stationary state, neural filters are orthonormal:
\begin{align}\label{orthF}
{\bf F}{\bf F}^\top &= {\bf I}_m.
\end{align}
\end{Th}
\begin{proof}First, we substitute for ${\bf F}$ (but not for ${\bf F^\top}$) its definition \eqref{QF}: ${\bf F}{\bf F}^\top = \left({\bf I}_m +  {\bf M} \right)^{-1}{\bf W}{\bf F}^\top$. Next, using Lemma \ref{lem1}, we substitute ${\bf W} {\bf F}^\top$ by $\left({\bf I}_m +  {\bf M} \right)$. The right hand side becomes $\left({\bf I}_m +  {\bf M} \right)^{-1}\left({\bf I}_m +  {\bf M} \right)={\bf I}_m$.
\end{proof}
\begin{rmk} Theorem \ref{Th1} implies that ${\rm rank (\bf F})=m$. 
\end{rmk}

\subsection{Neural filters and their relationship to the eigenspace of the covariance matrix}\label{sTh2}

How is the filter space related to the input? We partially answer this question in Theorem \ref{Th2}, using the following lemma:
\begin{lem}\label{lem2} In the stationary state, ${\bf F}^\top{\bf F}$ and ${\bf C}$ commute:
\begin{align}\label{FC}
 {\bf F}^\top{\bf F} {\bf C} &= {\bf C}{\bf F}^\top{\bf F}.
\end{align}
\end{lem}
\begin{proof} See Appendix \ref{A3}.
\end{proof}
Now we can state our second theorem.
\begin{Th}\label{Th2}
At the stationary state state, the filter space is an $m$-dimensional subspace in $\mathbb{R}^n$ that is spanned by some $m$ eigenvectors of the covariance matrix.
\end{Th}
\begin{proof}
Because ${\bf F}^\top{\bf F}$ and ${\bf C}$ commute (Lemma \ref{lem2}), they must share the same eigenvectors. Equation \eqref{orthF}  of Theorem \ref{Th1} implies that $m$ eigenvalues of ${\bf F}^\top{\bf F}$ are unity and the rest are zero. Eigenvectors associated with unit eigenvalues span the  rowspace of ${\bf F}$\footnote{If this fact is not familiar to the reader, we recommend Strang's \citep{strang2009} discussion of Singular Value Decomposition.} and are identical to some $m$ eigenvectors of {\bf C}.
\end{proof}

Which $m$ eigenvectors of {\bf C} span the filter space? To show that these are the eigenvectors corresponding to the largest eigenvalues of ${\bf C}$, we perform a linear stability analysis around the stationary point and show that any other combination would be unstable.

\subsection{Linear stability requires neural filters to span a principal subspace}\label{sTh3}

The strategy here is to perturb ${\bf F}$ from its equilibrium value and show that the perturbation is linearly stable only if the row space of ${\bf F}$ is the space spanned by the eigenvectors corresponding to the $m$ highest eigenvalues of ${\bf C}$. To prove this result, we will need two more lemmas.
\begin{lem}\label{lem3}  Let ${\bf H}$ be an $m\times n$ real matrix with orthonormal rows and ${\bf G}$ is an $(n-m) \times n$ real matrix with orthonormal rows, whose rows are chosen to be orthogonal to the rows of ${\bf H}$. Any ${n\times m}$ real matrix ${\bf Q}$ can be decomposed as:
\begin{align*}
{\bf Q} = {\bf A} \,{\bf H} + {\bf S}\, {\bf H}+ {\bf B}\, {\bf G},
\end{align*}
where ${\bf A}$ is an $m\times m$ skew-symmetric matrix, ${\bf S}$ is an $m\times m$ symmetric matrix and ${\bf B}$ is an $m\times(n-m)$ matrix. 
\end{lem}
\begin{proof} Define ${\bf B} :=  {\bf Q}\, {\bf G}^\top$, ${\bf A} :=  \frac 12 \left({\bf Q}\, {\bf H}^\top-{\bf H}\,{\bf Q}^\top\right)$ and ${\bf S} :=  \frac 12 \left({\bf Q}\, {\bf H}^\top+{\bf H}\,{\bf Q}^\top\right)$. Then, ${\bf A} \,{\bf H} + {\bf S}\, {\bf H}+ {\bf B}\, {\bf G}={\bf Q}\left({\bf H}^\top{\bf H}+{\bf G}^\top{\bf G}\right)={\bf Q}$.
\end{proof}

Let's denote an arbitrary  perturbation of ${\bf F}$ as $\delta{\bf F}$, where a small parameter is implied. We can use Lemma \ref{lem3} to decompose $\delta{\bf F}$ as
\begin{align}\label{dF}
\delta{\bf F} = \delta {\bf A} \,{\bf F} + \delta {\bf S}\, {\bf F}+\delta {\bf B}\, {\bf G},
\end{align}
where the rows of ${\bf G}$ are orthogonal to the rows of ${\bf F}$. Skew-symmetric ${\bf \delta A}$ corresponds to rotations of filters within the filter space, i.e. it keeps neural filters orthonormal. Symmetric ${\bf \delta S}$ keeps the filter space invariant but destroys orthonormality. $\delta {\bf B}$ is a perturbation that takes the neural filters outside of the filter space. 

Next, we calculate how $\delta {\bf F}$ evolves under the learning rule, i.e. $\left<\Delta \delta{\bf F}\right>$. 
\begin{lem}\label{lem4} A perturbation to the stationary state has the following evolution under the learning rule to linear order in perturbation and linear order in $T^{-1}$:
\begin{multline}\label{ddf}
\left<\Delta \delta{F}_{ij}\right> =\frac{1}{T}\sum_k \frac{\left({\bf I}_m+\bf{M}\right)^{-1}_{ik}}{ \left<y^2_k\right>}\left[ \sum_l \delta F_{kl}C_{lj}-\sum_{lpr}\delta F_{kl}F_{rp}C_{lp}F_{rj} \right. \\
\left.- \sum_{lpr} F_{kl}\delta F_{rp}C_{lp}F_{rj}\right]-\frac 1T \delta F_{ij}.
\end{multline}
\end{lem}
\begin{proof} Proof is provided in Appendix \ref{A4}.
\end{proof}

Now, we can state our main result in the following theorem:
\begin{Th}\label{Th3}
The stationary state of neuronal filters F is stable, in large-$T$ limit, only if the $m$ dimensional filter space is spanned by the eigenvectors of the covariance matrix corresponding to the $m$ highest eigenvectors.
\end{Th}
\begin{proof}[Proof Sketch] Full proof is given in Appendix \ref{A5}. Here we sketch the proof.

To simplify our analysis, we choose a specific ${\bf G}$ in Lemma \ref{lem3} without losing generality. Let ${\bf v}^{1,\ldots,n}$ be eigenvectors of ${\bf C}$ and $v^{1,\ldots,n}$ be corresponding eigenvalues, labeled so that the first $m$ eigenvectors span the row space of ${\bf F}$ (or filter space). We choose rows of ${\bf G}$ to be the remaining eigenvectors, i.e. ${\bf G}':=[{\bf v}^{m+1},\ldots,{\bf v}^{n}]$.

By extracting the evolution of components of $\delta {\bf F}$ from \eqref{ddf} using \eqref{dF}, we are ready to state the conditions under which perturbations of ${\bf F}$ are stable. Mutlipying \eqref{ddf} on the right by ${\bf G}^\top$ gives the evolution of $\delta {\bf B}$:
\begin{align*}
 \left<\Delta \delta{B}_{i}^j\right> = \sum_k P^j_{ik}\delta B_{k}^j\quad \text{where}\quad P^j_{ik}\equiv \frac{1}{T}\left(\frac{\left({\bf I}_m+\bf{M}\right)^{-1}_{ik}}{\left<y^2_k\right>}v^{j+m} -\delta_{ik}\right). 
\end{align*}
Here we changed our notation to $\delta B_{kj}=\delta B_{k}^j$ to make it explicit that for each $j$ we have one matrix equation. These equations are stable when all eigenvalues of all ${\bf P}^j$ are negative, which requires as shown in the Appendix \ref{A5}:
\begin{align*}
\left\lbrace v^{1}, \ldots,v^{m}\right\rbrace >  \left\lbrace v^{m+1}, \ldots,v^{n}\right\rbrace.
\end{align*}
This result proves that the perturbation is stable only if the filter space is identical to the space spanned by eigenvectors corresponding to the $m$ highest eigenvalues of ${\bf C}$. 

It remains to analyze the stability of $\delta {\bf A}$ and $\delta {\bf S}$ perturbations. Multiplying \eqref{ddf} on the right by ${\bf F}^\top$ gives,
 \begin{align*}
\left<\Delta \delta{A}_{ij}\right> &= 0  \qquad {\rm and} \qquad \left<\Delta \delta{S}_{ij}\right> = \frac{-2}{T}\delta{S}_{ij}.
 \end{align*}
$\delta {\bf A}$ perturbation, which rotates neural filters, does not decay. This behavior is inherently related to the discussed symmetry of the strain cost function \eqref{strain} with respect to left rotations of the {\bf Y} matrix. Rotated {\bf y} vectors are obtained from the input by rotated neural filters and hence $\delta {\bf A}$ perturbation does not affect the cost. On the other hand, $\delta {\bf S}$ destroys orthonormality and these perturbations do decay, making the orthonormal solution stable. 
\end{proof}

To summarize our analysis, if the dynamics converges to a stationary state, neural filters form an orthonormal basis of the principal subspace. 

\section{Numerical simulations of the asynchronous network}\label{secNum}

Here, we simulate the performance of the network with asynchronous updates, \eqref{dynamics} and \eqref{hah}, on synthetic data. The data were generated by a colored Gaussian process with an arbitrarily chosen ``actual" covariance matrix. We choose the number of input channels, $n=64$, and the number of output channels, $m=4$. In the input data, the ratio of the power in first 4 principal components to the power in remaining 60 components was 0.54. 
 ${\bf W}$ and ${\bf M}$ were initialized randomly, the step size of synaptic updates were initialized to $1/D_{0,i}=0.1$. Coordinate descent step is cycled over neurons until magnitude of change in ${\bf y}_T$ in one cycle is less than $10^{-5}$ times the magnitude of ${\bf y}_T$. 

We compared the performance of the asynchronous updates network, \eqref{dynamics} and \eqref{hah}, with two previously proposed networks, APEX \citep{kung1990neural,kung1994adaptive} and F\"oldiak's  \citep{foldiak1989adaptive}, on the same dataset, Figure \ref{Fig2}. APEX network uses the same Hebbian/anti-Hebbian learning rules for synaptic weights, but the architecture is slightly different in that the lateral connection matrix, ${\bf M}$, is lower triangular. F\"oldiak's network has the same architecture as ours, Figure \ref{Fig1}B, and the same learning rules for feedforward connections. However,  the learning rule for lateral connections is $\Delta M_{ij}\propto y_i y_j$, unlike \eqref{hah}. For the sake of fairness, we applied the same adaptive step size procedure for all networks. As in \eqref{hah}, the stepsize for each neuron $i$ at time $T$ was $1/D_{T+1,i}$, with $D_{T+1,i} = D_{T,i} + y_{T,i}^2$.  In fact, such learning rate has been recommended and argued to be optimal for the APEX network \citep{diamantaras1996principal,diamantaras2002neural}, see also footnote \ref{StepSize}. 

To quantify the performance of these algorithms, we used three different metrics. First is the strain cost function, \eqref{strain}, normalized by $T^2$, Figure \ref{Fig2}A. Such normalization is chosen because the minimum value of offline strain cost is equal to the power contained in the eigenmodes beyond the top $m$: $T^2\sum_{k=m+1}^n\left(v^k\right)^2$, where $\lbrace v^1, \ldots,v^n\rbrace$ are eigenvalues of sample covariance matrix ${\bf C}_T$ \citep{cox2000multidimensional,mardia1980multivariate}. For each of the three networks, as expected, the  strain cost rapidly drops towards its lower bound. As our network was derived from the minimization of the strain cost function, it is not surprising that the cost drops faster than in the other two.

Second metric quantifies the deviation of the learned subspace from the actual principal subspace. At each $T$, the deviation is $\left\Vert{\bf F}_T^\top{\bf F}_T-{\bf V}^\top{\bf V}\right\Vert_F^2$, where ${\bf V}$ is a $m\times n$ matrix whose rows are the principal eigenvectors, ${\bf V}^\top{\bf V}$ is the projection matrix to the principal subspace,  ${\bf F}_T$ is defined the same way for APEX and F\"oldiak networks as ours and ${\bf F}_T^\top{\bf F}_T$ is the learned estimate of the projection matrix to the principal subspace. Such deviation rapidly falls for each network confirming that all three algorithms learn the principal subspace, Figure \ref{Fig2}B. Again, our algorithm extracts the principal subspace faster than the other two networks. 

Third metric measures the degree of non-orthonormality among the computed neural filters. At each $T$: $\left\Vert{\bf F}_T{\bf F}_T^\top-{\bf I}_m\right\Vert_F^2$. Non-orthonormality error quickly drops for all networks, confirming that neural filters converge to orthonormal vectors, Figure \ref{Fig2}C. Yet again, our network orthonormalizes neural filters much faster than the other two networks. 

\begin{figure}
\centering
\includegraphics[width=\textwidth]{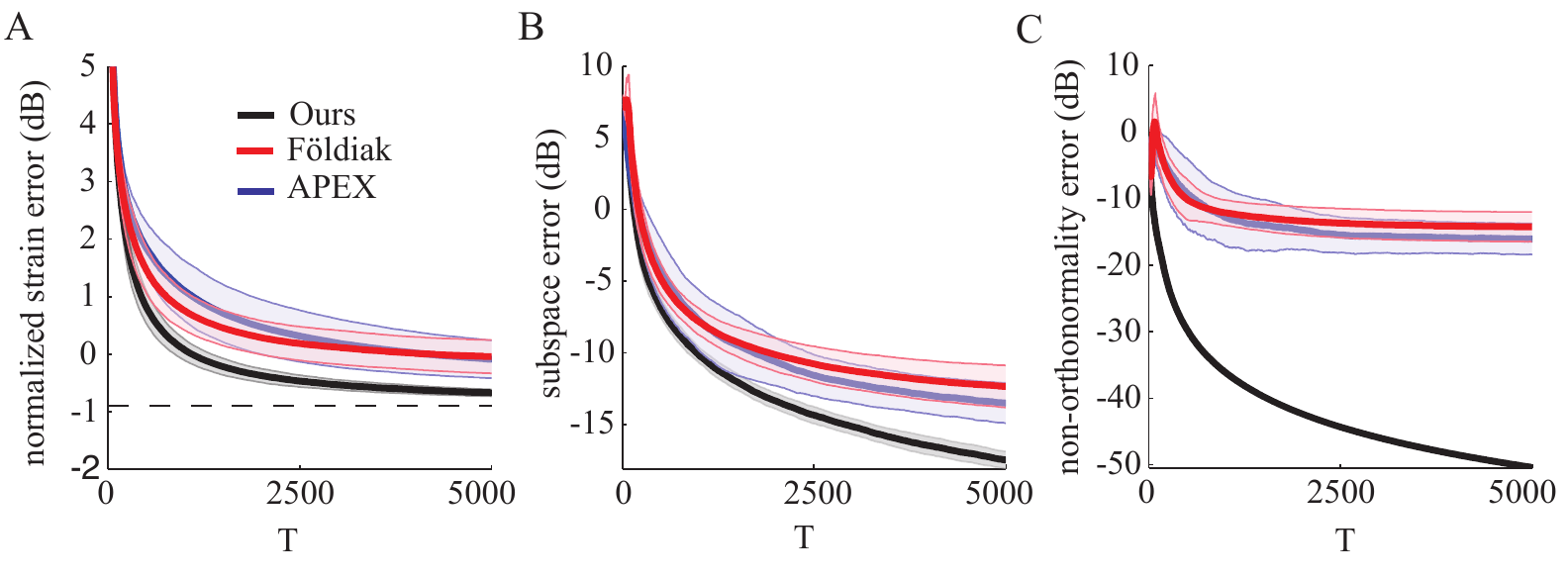}
  \caption{\label{Fig2}Performance of the asynchronous neural network compared with existing algorithms. Each algorithm was applied to 40 different random data sets drawn from the same Gaussian statistics, described in text. Weight initializations were random. Solid lines indicate means and shades indicate standard deviations across 40 runs.  All errors are in decibells (dB). For formal metric definitions, see text. {\bf A.} Strain error as a function of data presentations. Dotted line is the best error in batch setting, calculated using eigenvalues of the actual covariance matrix. {\bf B.} Subspace error as a function of data presentations. {\bf C.} Non-orthonormality error as a function of data presentations.}
\end{figure}

\section{Subspace tracking using a neural network with local learning rules}

We have demonstrated that our network learns a linear subspace of streaming data generated by a stationary distribution. But what if the data are generated by an evolving distribution and we need to track the corresponding linear subspace? Using the algorithm \eqref{hah} would be suboptimal because the learning rate is adjusted to effectively ``remember" the contribution of all the past data points.  

A natural way to track an evolving subspace is to ``forget" the contribution of older data points \citep{yang1995projection}. In this Section, we derive an algorithm with ``forgetting" from a principled cost function where errors in the similarity of old data points are discounted: 
\begin{align}\label{discount0}
{{\bf y}_T}= \mathop {\arg \min }\limits_{{\bf y}_T} \sum_{t=1}^T\sum_{t'=1}^{T}\beta^{2T-t-t'}\left({\bf x}_t^\top{\bf x}_{t'}-{\bf y}_t^\top{\bf y}_{t'}\right)^2. 
\end{align}
where $\beta$ is a discounting factor $0\leq\beta\leq1$ with $\beta=1$ corresponding to our original algorithm \eqref{Frob}. The effective time scale of ``forgetting" is:
\begin{align}
\tau := -1/\ln \beta.
\end{align} 

By introducing a $T \times T$-dimensional diagonal matrix $\beta_T$ with diagonal elements $\beta_{T,{ii}} = \beta^{T-i}$ we can rewrite \eqref{discount0} in a matrix notation:  
\begin{align}\label{discount1}
{{\bf y}_T} &=  \mathop {\arg \min }\limits_{{\bf y}_T} \left\Vert \beta_T^\top{\bf X}^\top{\bf X}\beta_T-\beta_T^\top{\bf Y}^\top{\bf Y}\beta_T\right\Vert_F^2.
\end{align}
A similar discounting was used in  \citep{yang1995projection} to derive subspace tracking algorithms from the representation error cost function, \eqref{PCA}.

To derive an online algorithm to solve \eqref{discount1} we follow the same steps as before. By keeping only the terms that depend on current output ${\bf y}_T$ we get:
\begin{multline} \label{onlineFullDisc}
{{\bf y}_T}  =\mathop {\arg \min }\limits_{{\bf y}_T} \left[  - 4{{\bf x}^\top_T}\left( {\sum\limits_{t = 1}^{T - 1} \beta^{2(T-t)}{{\bf x}_t{{\bf y}^\top_t}} } \right){\bf y}_T + 2{{\bf y}^\top_T}\left( {\sum\limits_{t = 1}^{T - 1}\beta^{2(T-t)} {{\bf y}_t{{\bf y}^\top_t}} } \right){\bf y}_T  \right. \\ \left.- 2{{\left\| {{\bf x}_T} \right\|}^2}{{\left\| {{\bf y}_T} \right\|}^2} + {{\left\| {{\bf y}_T} \right\|}^4} \right].
\end{multline}
In \eqref{onlineFullDisc}, provided that past input-input and input-output outer products are not forgotten for a sufficiently long time, i.e. $\tau>>1$, the first two terms dominate over the last two for large $T$. After dropping the last two terms we arrive at: 
\begin{align}\label{onlineStrainDisc}
{{\bf y}_T} =\mathop {\arg \min }\limits_{{\bf y}_T} \left[ { - 4{{\bf x}^\top_T}\left( {\sum\limits_{t = 1}^{T - 1}\beta^{2(T-t)} {{\bf x}_t{{\bf y}^\top_t}} } \right){\bf y}_T + 2{{\bf y}^\top_T}\left( {\sum\limits_{t = 1}^{T - 1} \beta^{2(T-t)} {{\bf y}_t{{\bf y}^\top_t}} } \right){\bf y}_T  } \right].
\end{align}

As in the non-discounted case, minimization of the discounted online CMDS cost function by coordinate descent \eqref{onlineStrainDisc} leads to a neural network with asynchronous updates,
\begin{align}\label{asynch}
{y_{T,i}} \leftarrow \sum_{j=1}^n W^\beta_{T,ij} x_{T,j} - \sum_{j=1}^m M^\beta_{T,ij} y_{T,j}, 
\end{align}
and by a Jacobi iteration - to a neural network with synchronous updates,
\begin{align}\label{synch}
{\bf y}_T \leftarrow {\bf W}^\beta_T{\bf x}_T-{\bf M}^\beta_T{\bf y}_T,
\end{align}
with synaptic weight matrices in both cases given by:
\begin{align}\label{WMdisc}
W_{T,ij}^\beta = \frac{{\sum\limits_{t = 1}^{T - 1}\beta^{2(T-t)} {y_{t,i}^{}x_{t,j}^{}} }}{{\sum\limits_{t = 1}^{T - 1} \beta^{2(T-t)}{y_{t,i}^2} }},\qquad  M_{T,i,j \ne i}^\beta= \frac{{\sum\limits_{t = 1}^{T - 1} \beta^{2(T-t)}{y_{t,i}^{}y_{t,j}^{}} }}{{\sum\limits_{t = 1}^{T - 1} \beta^{2(T-t)} {y_{t,i}^2} }},\qquad M_{T,ii}^\beta = 0.
\end{align}

Finally, we rewrite \eqref{WMdisc} in a recursive form. As before, we introduce a scalar variable ${D}_{T,i}^\beta$ representing the discounted cumulative activity of a neuron $i$ up to time $T-1$, 
\begin{align}\label{YDisc}
{D}_{T,i}^\beta =\sum\limits_{t = 1}^{T - 1} \beta^{2(T-t-1)} {y_{t,i}^2}.
\end{align}
Then, the recursive updates are:
\begin{align}\label{hahDisc}
D_{T+1,i}^\beta &\leftarrow \beta^2D_{T,i}^\beta+ y_{T,i}^2\nonumber\\
W_{T+1,ij}^\beta &\leftarrow W_{T ,ij}^\beta + y_{T,i}\left( {x_{T,j} - W_{T,ij}^\beta y_{T,i}} \right)/D_{T+1,i}^\beta \nonumber\\
M_{T+1,i,j \ne i}^\beta &\leftarrow M_{T,ij}^\beta + y_{T,i}\left( {y_{T,j} - M_{T,ij}^\beta y_{T,i}^{}} \right)/D_{T+1,i}^\beta.
\end{align}
These updates are local and almost identical to the original updates \eqref{hah} except the ${D}_{T+1,i}^\beta$ update, where the past cumulative activity is discounted by $\beta^2$. For suitably chosen $\beta$, the learning rate, $1/{D}_{T+1,i}^\beta$, stays sufficiently large even at large-$T$, allowing the algorithm to react to changes in data statistics. 

As before, we have a two-phase algorithm for minimizing the discounted online CMDS cost function \eqref{onlineStrainDisc}. For each data presentation, first the neural network dynamics is run using either \eqref{asynch} or \eqref{synch} until the dynamics converges to a fixed point. In the second step, synaptic weights are updated using \eqref{hahDisc}.

In Figure \ref{Fig4}, we present the results of a numerical simulation of our subspace tracking algorithm with asynchronous updates  similar to that in Section \ref{secNum} but for nonstationary synthetic data. The data are drawn from two different Gaussian distributions: from $T=1$ to $T=2500$ - with covariance ${\bf C}_1$,  and from $T=2501$ to $T=5000$ - with covariance ${\bf C}_2$. We ran our algorithm with 4 different $\beta$ factors,  $\beta=0.998, 0.995, 0.99, 0.98$ ($\tau = 499.5, 199.5, 99.5, 49.5$).

We evaluate the subspace tracking performance of the algorithm using a modification of the subspace error metric introduced in Section \ref{secNum}. From $T=1$ to $T=2500$ the error is $\left\Vert{\bf F}_T^\top{\bf F}_T-{\bf V}_1^\top{\bf V}_1\right\Vert_F^2$, where ${\bf V}_1$ is a $m\times n$ matrix whose rows are the principal eigenvectors of ${\bf C}_1$. From $T=2501$ to $T=5000$ the error is $\left\Vert{\bf F}_T^\top{\bf F}_T-{\bf V}_2^\top{\bf V}_2\right\Vert_F^2$, where ${\bf V}_2$ is a $m\times n$ matrix whose rows are the principal eigenvectors of ${\bf C}_2$. Figure \ref{Fig4}A plots this modified subspace error.  Initially, the subspace error decreases, reaching lower values with higher $\beta$. Higher $\beta$ allows for smaller learning rates allowing a fine-tuning of the neural filters and hence lower error. At $T=2501$, a sudden jump is observed corresponding to the change in principal subspace. The network rapidly corrects its neural filters to project to the new principal subspace and the error falls to before jump values. It is interesting to note that higher $\beta$ now leads to a slower decay due to extended memory in the past. 

We also quantify the degree of non-orthonormality of neural filters using the non-orthonormality error defined in Section \ref{secNum}. Initially, the non-orthonormality error decreases, reaching lower values with higher $\beta$. Again, higher $\beta$ allows for smaller learning rates allowing a fine-tuning of the neural filters. At $T=2501$, an increase in orthonormality error is observed as the network is adjusting its neural filters. Then, the error falls to before change values, with higher $\beta$ leading to a slower decay due to extended memory in the past.

\begin{figure}
\centering
\includegraphics{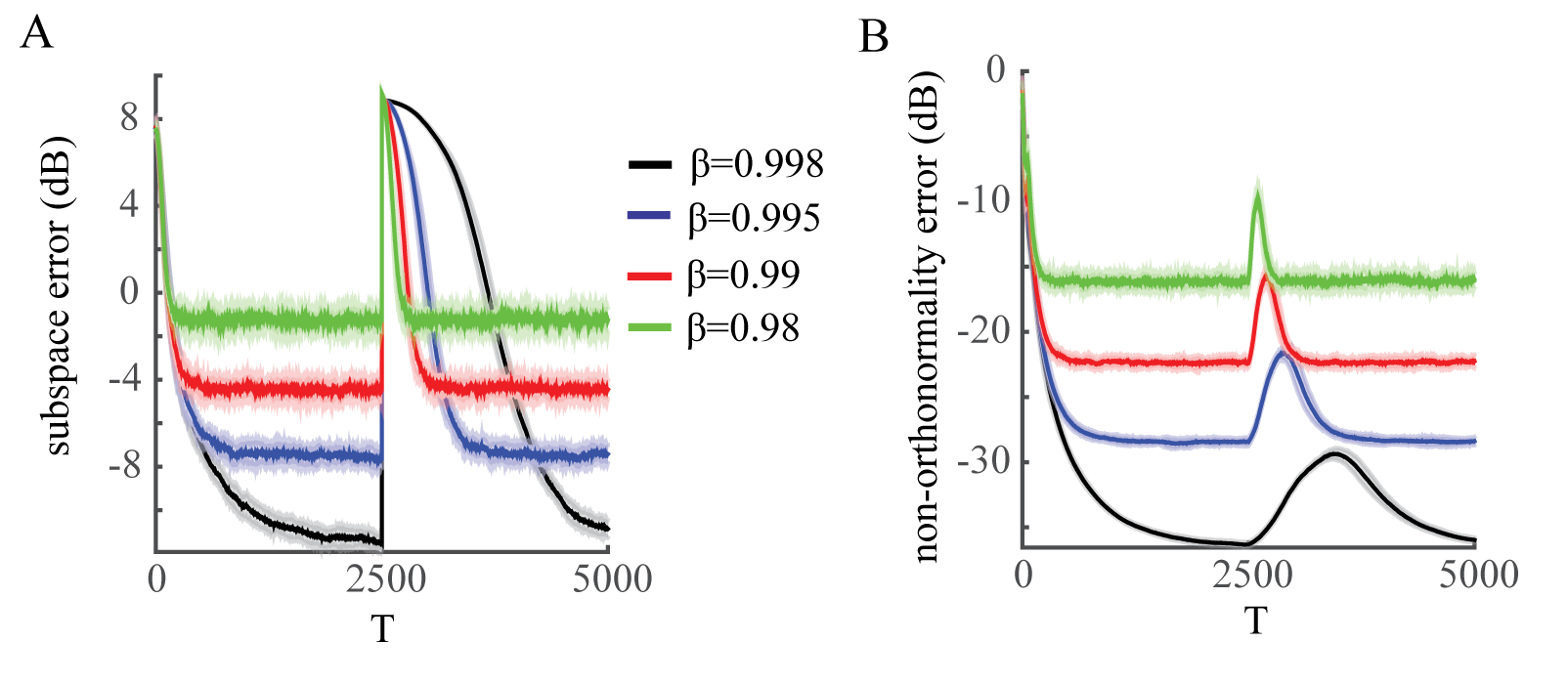}
\caption{\label{Fig4} Performance of the subspace tracking asynchronous neural network with nonstationary data. The algorithm with different $\beta$ factors was applied to 40 different random data sets drawn from the same nonstationary statistics, described in text. Weight initializations were random. Solid lines indicate means and shades indicate standard deviations.  All errors are in decibells (dB). For formal metric definitions, see text.  {\bf A.} Subspace error as a function of data presentations. {\bf B.} Non-orthonormality error as a function of data presentations.}
\end{figure}

\section{Discussion}

In this paper, we made a step towards a mathematically rigorous model of neuronal dimensionality reduction satisfying more biological constraints than was previously possible. By starting with the CMDS cost function \eqref{strain}, we derived a single-layer neural network of linear units using only local learning rules. Using a local stability analysis, we showed that our algorithm finds a set of orthonormal neural filters and projects the input data stream to its principal subspace. We showed that with a small modification in learning rate updates, the same algorithm performs subspace tracking. 

Our algorithm finds the principal subspace, but not necessarily the principal components themselves. This is not a weakness since both the representation error cost \eqref{PCA} and CMDS cost \eqref{strain} are minimized by projections to principal subspace and finding the principal components is not necessary. 

Our network is most similar to F\"oldiak's network \citep{foldiak1989adaptive}, which learns feedforward weights by a Hebbian Oja rule and the all-to-all lateral weights by an anti-Hebbian rule. Yet, the functional form of the anti-Hebbian learning rule in  F\"oldiak's network, $\Delta M_{ij}\propto y_i y_j$, is different from ours \eqref{hah} resulting in the following interesting differences: 1) Because the synaptic weight update rules in F\"oldiak's network are symmetric, if the weights are initialized symmetric, i.e. $M_{ij}=M_{ji}$, and learning rates are identical for lateral weights, they will stay symmetric. As mentioned above, such symmetry does not exist in our network (\eqref{hah} and \eqref{main2}). 2) While in F\"oldiak's network neural filters need not be orthonormal \citep{foldiak1989adaptive,leen1991}, in our network they will be (Theorem \ref{Th1}). 3)  In F\"oldiak's network output units are decorrelated \citep{foldiak1989adaptive}, since in its stationary state $\left<y_i y_j\right>=0$. This need not be true in our network. Yet, correlations among output units do not necessarily mean that information in the output about the input is reduced\footnote{As pointed before \citep{linsker1988,plumbley1993,plumbley1995,kung2014kernel}, PCA maximizes mutual information between a Gaussian input, ${\bf x}$, and an output, ${\bf y}= {\bf F}{\bf x}$, such that rows of ${\bf F}$ have unit norms. When rows of ${\bf F}$ are principal eigenvectors, outputs are principal components and are uncorrelated. However, the output can be multiplied by a rotation matrix, ${\bf Q}$, and mutual information is unchanged, ${\bf y}' ={\bf Q} {\bf y}= {\bf Q}{\bf F}{\bf x}$. ${\bf y}'$ is now a correlated Gaussian and $ {\bf Q}{\bf F}$ still has rows with unit norms. Therefore, one can have correlated outputs with maximal mutual information between input and output, as long as rows of ${\bf F}$ span the principal subspace.}.

Our network is similar to the APEX network \citep{kung1990neural} in the functional form of both the feedforward and the lateral weights. However the network architecture is different because the APEX network has a lower-triangular lateral connectivity matrix. Such difference in architecture leads to two interesting differences in the APEX network operation  \citep{diamantaras1996principal}: 1) The outputs converge to the principal components. 2) Lateral weights decay to zero and neural filters are the feedforward weights. In our network lateral weights do not have to decay to zero and neural filters depend on both the feedforward and lateral weights \eqref{QF}.

In numerical simulations, we observed that our network is faster than F\"oldiak's and  APEX networks in minimizing the strain error, finding the principal subspace and orthonormalizing neural filters. This result demonstrates the advantage of our principled approach compared to heuristic learning rules.

Our choice of coordinate descent to minimize the cost function in the activity dynamics phase allowed us to circumvent problems associated with matrix inversion: ${\bf y} \leftarrow ({\bf I}_m+{\bf M})^{-1}{\bf W}{\bf x}$. Matrix inversion causes problems for neural network implementations because it is a non-local operation.  In the absence of a cost function, F\"oldiak suggested to implement matrix inversion by iterating ${\bf y}\leftarrow {\bf W}{\bf x} - {\bf M}{\bf y}$ until convergence \citep{foldiak1989adaptive}. We derived a similar algorithm using Jacobi iteration. However, in general, such iterative schemes are not guaranteed to converge \citep{hornik1992convergence}. Our coordinate descent algorithm is almost always guaranteed to converge because the cost function in the activity dynamics phase \eqref{onlineStrain} meets the criteria in \citep{luo1991on}.  

Unfortunately, our treatment still suffers from the problem common to most other biologically plausible neural networks \citep{hornik1992convergence}: a complete global convergence analysis of synaptic weights is not yet available. Our stability analysis is local in the sense that it starts by assuming that the synaptic weight dynamics has reached a stationary state and then proves that perturbations around the stationary state are stable. We have not made a theoretical statement on whether this state can ever be reached or how fast such a state can be reached. Global convergence results using stochastic approximation theory  are available for the single-neuron Oja rule \citep{oja1985stochastic}, its nonlocal generalizations \citep{plumbley1995} and the APEX rule \citep{diamantaras1996principal}, however applicability of stochastic approximation theory was questioned recently \citep{zufiria2002}. Even though a neural network implementation is unknown, Warmuth \& Kuzmin's  online PCA algorithm stands out as the only algorithm for which a regret bound has been proved \citep{warmuth2008randomized}. An asymptotic dependence of regret on time can also be interpreted as convergence speed. 

This paper also contributes to MDS literature by applying CMDS method to streaming data. However, our method has limitations in that to derive neural algorithms we used the strain cost \eqref{strain} of CMDS. Such cost is formulated in terms of similarities, inner products to be exact, between pairs of data vectors and allowed us to consider a streaming setting where a data vector is revealed at a time. In the most general formulation of MDS pairwise dissimilarities between data instances are given rather than data vectors themselves or similarities between them \citep{cox2000multidimensional,mardia1980multivariate}. This generates two immediate problems for a generalization of our approach: 1) A mapping to the strain cost function \eqref{strain} is only possible if the dissimilarites are Euclidean distances (footnote \ref{dissimilarity}). In general, dissimilarities do not need be Euclidean or even metric distances \citep{cox2000multidimensional,mardia1980multivariate} and one cannot start from the strain cost \eqref{strain} for derivation of a neural algorithm. 2) In the streaming version of the general MDS setting, at each step, dissimilarities between the current and all past data instances are revealed, unlike our approach where the data vector itself is revealed. It is a challenging problem for future studies to find neural implementations in such generalized setting.

The online CMDS cost functions \eqref{onlineStrain} and \eqref{onlineStrainDisc} should be valuable for subspace learning and tracking applications where biological plausibility is not a necessity. Minimization of such cost functions could be performed much more efficiently in the absence of constraints imposed by biology\footnote{For example, matrix equation \eqref{anMin} could be solved by a conjugate gradient descent method instead of iterative methods. Matrices that keep input-input and output-output correlations in \eqref{anMin} can be calculated recursively, leading to a truly online method}. It remains to be seen how the algorithms presented in this paper and their generalizations compare to state-of-the-art online subspace tracking algorithms from machine learning literature \citep{cichocki2002adaptive}.

Finally, we believe that formulating the cost function in terms of similarities supports the possibility of representation invariant computations in neural networks.

\subsubsection*{Acknowledgments}

We are grateful to L. Greengard, S. Seung and M. Warmuth for helpful discussions.

\appendix
\section{Appendix}

\subsection{Alternative derivation of an asynchronous network}\label{altAsynch}

Here, we solve the system of equations \eqref{anMin}  iteratively \citep{strang2009}. First, we split the output covariance matrix that appears on the left-hand side of \eqref{anMin} into its diagonal component ${\bf D}_T$, a strictly upper triangular matrix ${\bf U}_T$ and a strictly lower triangular matrix ${\bf L}_T$:
\begin{align}\label{DUL}
{\sum\limits_{t = 1}^{T - 1} {{\bf y}_t{{\bf y}^\top_t}} } = {\bf D}_T + {\bf U}_T + {\bf L}_T.
\end{align}
Substituting this into \eqref{anMin} we get:
\begin{align}\label{master}
\left({\bf D}_T + \omega {\bf L}_T\right){\bf y}_T  = \left(\left(1-\omega\right) {\bf D}_T -\omega {\bf U}_T \right) {\bf y}_T+ \omega \left( {\sum\limits_{t = 1}^{T - 1} {{\bf y}_t{{\bf x}^\top_t}} } \right){{\bf x}_T} ,
\end{align}
where $\omega$ is a parameter. We solve \eqref{anMin} by iterating
\begin{align}\label{masterIter}
{\bf y}_T  \longleftarrow \left({\bf D}_T + \omega {\bf L}_T\right)^{-1}\left[\left(\left(1-\omega\right) {\bf D}_T -\omega {\bf U}_T \right) {\bf y}_T+ \omega \left( {\sum\limits_{t = 1}^{T - 1} {{\bf y}_t{{\bf x}^\top_t}} } \right){{\bf x}_T}\right] ,
\end{align}
until convergence.  If symmetric ${\sum\limits_{t = 1}^{T - 1} {{\bf y}_t{{\bf y}^\top_t}} }$ is positive definite, the convergence is guaranteed for $0<\omega<2$ by the Ostrowski-Reich theorem \citep{reich1949,ostrowski1954}. When $\omega=1$ the iteration \eqref{masterIter} corresponds to the Gauss-Seidel method and when $\omega>1$ - to the succesive overrelaxation method. The choice of $\omega$ for fastest convergence depends on the problem, and we will not explore this question here. However, values around 1.9 are generally recommended \citep{strang2009}.

Because in \eqref{master} the matrix multiplying ${\bf y}_T$ on the left is lower triangular and on the right is upper triangular, the iteration \eqref{masterIter} can be performed component-by-component \citep{strang2009}:
\begin{align}
\label{solution1Alt}
{y_{T,i}} \longleftarrow \left(1-\omega \right){y_{T,i}}  +\omega \frac{{\sum\limits_k {\left( {\sum\limits_{t = 1}^{T - 1} {y_{t,i}^{}x_{t,k}^{}} } \right)x_{T,k}^{}} }}{{\sum\limits_{t = 1}^{T - 1} {y_{t,i}^2} }} - \omega \frac{{\sum\limits_{j \ne i} {\left( {\sum\limits_{t = 1}^{T - 1} {y_{t,i}^{}y_{t,j}^{}} } \right)y_{T,j}^{}} }}{{\sum\limits_{t = 1}^{T - 1} {y_{t,i}^2} }}.
\end{align}
Note that $y_{T,i}$ is replaced with its new value before moving to the next component.

This algorithm can be implemented in a neural network 
\begin{align}\label{dynamicsAlt}
{y_{T,i}} \leftarrow  \left(1-\omega \right){y_{T,i}} +\omega \sum_{j=1}^n W_{T,ij} x_{T,j} - \omega \sum_{j=1}^m M_{T,ij} y_{T,j}, 
\end{align}
where ${\bf W}_T$ and ${\bf M}_T$, as defined in \eqref{WM},  represent the synaptic weights of feedforward and lateral connections respectively. The case of $\omega < 1$ can be implemented by a leaky integrator neuron. The $\omega = 1$ case corresponds to our original asynchronous algorithm, except that now updates are performed in a particular order. For the $\omega>1$ case, which may converge faster, we do not see a biologically plausible implementation since it requires self-inhibition.

Finally, to express the algorithm in a fully online form we rewrite \eqref{WM} via recursive updates, resulting in \eqref{hah}.

\subsection{Proof of Lemma \ref{lem2}}\label{A3}

\begin{proof}[Proof of Lemma 2] In our derivation below, we use results from equations \eqref{QF}, \eqref{main1} and \eqref{main2} of the main text.
\begin{align*}
\left({\bf F}^\top{\bf F} {\bf C}\right)_{ij} &= \sum_{kl} F_{ki} F_{kl} \left<x_lx_j\right> \nonumber \\
&=  \sum_{k} F_{ki}  \left<y_kx_j\right> {\tag{from \eqref{QF}}}\nonumber \\	
&=  \sum_{k} F_{ki}  \left<y_k^2\right>W_{kj} &{\tag{from \eqref{main1}}}\\
&=\sum_{kp} F_{ki}  \left<y_k^2\right>\left(M_{kp}+\delta_{kp}\right)F_{pj} &{\tag{from \eqref{QF}}}\nonumber \\
&=\sum_{kp} F_{ki}  \left<y_p^2\right>\left(M_{pk}+\delta_{pk}\right)F_{pj} &\tag{from \eqref{main2}}\\
&=\sum_{p} W_{pi}  \left<y_p^2\right>F_{pj} &\tag{from \eqref{QF}} \\
&=\sum_{p} \left<y_px_i\right>F_{pj} &\tag{from \eqref{main1}} \\
&=\sum_{pk} F_{pk}\left<x_kx_i\right>F_{pj} = \sum_{pk} \left<x_ix_k\right>F_{pk}F_{pj} =\left({\bf C}{\bf F}^\top{\bf F} \right)_{ij}. &\tag{from \eqref{QF}}
\end{align*}
\end{proof}

\subsection{Proof of Lemma \ref{lem4}}\label{A4}

Here we calculate how $\delta {\bf F}$ evolves under the learning rule, i.e. $\left<\Delta \delta{\bf F}\right>$ and derive equation \eqref{ddf}. 

First, we introduce some new notation to simplify our expressions. We define lateral synaptic weight matrix ${\bf M}$ with diagonals set to 1 as 
\begin{align}\label{Mhat}
{\bf \hat M} := {\bf I}_m + {\bf M}.
\end{align}
We use $\,\tilde{}\,$ to denote perturbed matrices
\begin{align}\label{PM}
\tilde {\bf F} &:= {\bf F} + \delta{\bf F}, & \tilde {\bf W} &:= {\bf W} + \delta{\bf W}, \nonumber \\  \tilde {\bf M} &:= {\bf M} + \delta{\bf M},  &\hat {\tilde {\bf  M}}&:= {\bf I} + \tilde {\bf M}= \hat {\bf M} + \delta{\bf M}.
\end{align}
Note that when the network is run with these perturbed synaptic matrices, for input ${\bf x}$, the network dynamics will settle to the fixed point
\begin{align}\label{Py}
\tilde{\bf y} = \hat {\tilde {\bf  M}}^{-1}\tilde {\bf W} {\bf x} = \tilde {\bf F} {\bf x},
\end{align}
which is different from the fixed point of the stationary network, ${\bf y} = \hat {\bf M}^{-1} {\bf W} {\bf x} =  {\bf F} {\bf x}$.

Now we can prove Lemma \ref{lem4}.
\begin{proof}[Proof of Lemma \ref{lem4}]
The proof goes in the following steps.
\begin{enumerate}[leftmargin=*]
\item Since our update rules are formulated in terms of ${\bf W}$ and ${\bf M}$, it will be helpful to express $\delta {\bf F}$ in terms of $\delta {\bf W}$ and $\delta {\bf M}$. The definition of ${\bf F}$, equation \eqref{QF}, gives us the desired relation:
\begin{align} \label{smallDelta}
(\delta {\bf \hat M}) {\bf F} + {\bf \hat M} (\delta {\bf F}) = \delta {\bf W}.
\end{align}
\item Next, we show that in the stationary state
\begin{align}\label{def}
\left<\Delta \delta {\bf F}\right> &= {\bf \hat M}^{-1}\left(\left<\Delta \delta {\bf W}\right>-\left<\Delta \delta {\bf M}\right>{\bf F}\right)+\mathcal{O}\left(\frac 1{T^2}\right).
\end{align}
\begin{proof}Average changes due to synaptic updates on both sides of \eqref{smallDelta} are equal: $\left<\Delta\left[(\delta {\bf \hat M}) {\bf F} + {\bf \hat M} (\delta {\bf F})\right]\right>  = \left<\Delta\delta {\bf W}\right> $. Noting that the unperturbed matrices are stationary, i.e. $\left<\Delta{\bf   M}\right> = \left<\Delta{\bf   F}\right> = \left<\Delta{\bf   W}\right>=0$, one gets $ \left<\Delta\delta{\bf   M}\right>{\bf F} + {\bf \hat M}\left<\Delta\delta{\bf F}\right> = \left<\Delta \delta{\bf W}\right>+\mathcal{O}\left({T^{-2}}\right)$, from which equation \eqref{def} follows.\end{proof}
\item Next step is to calculate $\left<\Delta \delta{\bf W}\right>$ and $\left<\Delta \delta{\bf M}\right>$ using the learning rule, in terms of matrices ${\bf W}$, ${\bf M}$, ${\bf C}$, ${\bf F}$ and $\delta{\bf F}$ and plug the result into \eqref{def}. This manipulation is going to give us the evolution of $\delta {\bf F}$ equation, \eqref{ddf}.

First, $\left<\Delta \delta{\bf W}\right>$ :
 \begin{align*}
 \left<\Delta \delta{W}_{ij}\right> &= \left<\Delta \tilde {W}_{ij}\right> \\
 &=\frac{1}{T\left<y^2_i\right>} \left(\left<\tilde y_ix_j\right>-\left<\tilde{y}_i^2\right>\tilde W_{ij}\right) \\
 &=\frac{1}{T\left<y^2_i\right>}\left(\sum_{k}\tilde F_{ik}\left<x_kx_j\right>-\sum_{kl}\tilde F_{ik}\tilde F_{il}\left<x_kx_l\right>\tilde W_{ij}\right) \tag{from \eqref{Py}}\\
 &=\frac{1}{T\left<y^2_i\right>}\left(\sum_{k}\tilde F_{ik}C_{kj}-\sum_{kl}\tilde F_{ik}\tilde F_{il}C_{kl}\tilde W_{ij}\right) \\
 &=\frac{1}{T\left<y^2_i\right>}\left(\sum_{k}F_{ik}C_{kj}-\sum_{kl}F_{ik}F_{il}C_{kl}W_{ij} +\sum_{k}\delta F_{ik}C_{kj}\right. \nonumber\\
 &\qquad \qquad \left.-2\sum_{kl}\delta F_{ik}F_{il}C_{kl}W_{ij}-\sum_{kl} F_{ik}F_{il}C_{kl}\delta W_{ij}  \right) \tag{from \eqref{PM}}\\
 &=\frac{1}{T\left<y^2_i\right>}\left( \sum_{k}\delta F_{ik}C_{kj}-2\sum_{kl}\delta F_{ik}F_{il}C_{kl}W_{ij}\right. \nonumber \\ &\qquad\qquad \left.-\sum_{kl} F_{ik}F_{il}C_{kl}\delta W_{ij}  \right).\tag{from \eqref{main1}} 
 \end{align*}
Next we calculate $\left<\Delta \delta{\bf M}\right>$ :
 \begin{align*}
 \left<\Delta \delta{M}_{ij}\right> &= \left<\Delta \tilde {\hat M}_{ij}\right> \\
 &= \frac{1}{T\left<y^2_i\right>}\left(\left<\tilde y_i \tilde y_j\right>-\left<{\tilde y}_i^2\right>\tilde M_{ij}\right)- \frac 1{D_i}\delta_{ij}\left<{\tilde y}_i^2\right> \tag{last term sets $\Delta { \tilde {\hat M}}_{ii}=0$}\nonumber \\
 &=\frac{1}{T\left<y^2_i\right>}\left(\sum_{kl}\tilde F_{ik} \tilde F_{jl}\left<x_kx_l\right>-\sum_{kl}\tilde F_{ik}\tilde F_{il}\left<x_kx_l\right>\tilde M_{ij}\right.\nonumber \\ &\qquad\qquad \left.-\delta_{ij}\sum_{kl}\tilde F_{ik}\tilde F_{il}\left<x_kx_l\right>\right) \tag{from \eqref{Py}} \\
  &=\frac{1}{T\left<y^2_i\right>}\left(\sum_{kl}\tilde F_{ik}\tilde F_{jl}C_{kl}-\sum_{kl}\tilde F_{ik}\tilde F_{il}C_{kl}\tilde M_{ij}-\delta_{ij}\sum_{kl}\tilde F_{ik} \tilde F_{il}C_{kl}\right) \\
 &=\frac{1}{T\left<y^2_i\right>}\left(\sum_{kl}F_{ik}F_{jl}C_{kl}-\sum_{kl}F_{ik}F_{il}C_{kl}M_{ij}-\delta_{ij}\sum_{kl}F_{ik}F_{il}C_{kl} \right. \nonumber\\
 &\qquad \qquad  +\sum_{kl}\delta F_{ik}F_{jl}C_{kl}+\sum_{kl}F_{ik}\delta F_{jl}C_{kl} -2\sum_{kl}\delta F_{ik}F_{il}C_{kl}M_{ij}  \nonumber \\
 &\qquad\qquad  \left.-\sum_{kl}F_{ik}F_{il}C_{kl}\delta M_{ij}-2\delta_{ij}\sum_{kl}\delta F_{ik}F_{il}C_{kl}\right) \tag{from \eqref{PM}}\\
 &=\frac{1}{T\left<y^2_i\right>}\left(\sum_{kl}\delta F_{ik}F_{jl}C_{kl}+\sum_{kl}F_{ik}\delta F_{jl}C_{kl} -2\sum_{kl}\delta F_{ik}F_{il}C_{kl}M_{ij}  \right. \nonumber \\
 &\qquad\qquad  \left. -\sum_{kl}F_{ik}F_{il}C_{kl}\delta M_{ij}-2\delta_{ij}\sum_{kl}\delta F_{ik}F_{il}C_{kl}\right). \tag{from \eqref{main2}}
 \end{align*}

Plugging these in equation \eqref{def}, we get
\begin{align*}
\left<\Delta \delta{ F}_{ij}\right> &= \sum_k \frac {\hat M^{-1}_{ik}}{T\left<y^2_k\right>}\left[ \sum_l \delta F_{kl}C_{lj}-2\sum_{lp}\delta F_{kl}F_{kp}C_{lp}W_{kj}-\sum_{lp}F_{kl}F_{kp}C_{lp}\delta W_{kj}\right. \nonumber\\
&\qquad\qquad\qquad - \sum_{lpr}\delta F_{kl}F_{rp}C_{lp}F_{rj} - \sum_{lpr} F_{kl}\delta F_{rp}C_{lp}F_{rj} \nonumber\\
&\qquad\qquad\qquad  +2\sum_{lpr} \delta F_{kl}F_{kp}C_{lp}M_{kr}F_{rj} +\sum_{lpr} F_{kl}F_{kp}C_{lp}\delta M_{kr}F_{rj}\nonumber\\
&\qquad\qquad\qquad \left.+2\sum_{lpr} \delta_{kr} \delta F_{kl}F_{kp}C_{lp} F_{rj}\right]+\mathcal{O}\left(\frac 1{T^2}\right).
\end{align*}
$M_{kr}$ and $\delta M_{kr}$ terms can be eliminated using the previously derived relations \eqref{QF} and \eqref{smallDelta}. This leads to a cancellation of some of the terms given above, and finally we have
\begin{multline*}
\left<\Delta \delta{F}_{ij}\right> = \sum_k \frac {\hat M^{-1}_{ik}}{T\left<y^2_k\right>} \left[ \sum_l \delta F_{kl}C_{lj}-\sum_{lpr}\delta F_{kl}F_{rp}C_{lp}F_{rj} \right. \\
\left.- \sum_{lpr} F_{kl}\delta F_{rp}C_{lp}F_{rj}- \sum_{lpr} F_{kl}F_{kp}C_{lp}\hat M_{kr}\delta F_{rj}\right]+\mathcal{O}\left(\frac 1{T^2}\right).
\end{multline*}
To proceed further, we note that:
\begin{align}\label{Qinv2}
{\left<y^2_k\right>} = \left({\bf F}{\bf C}{\bf F}^\top\right)_{kk},
\end{align}
which allows us to simplify the last term. Then, we get our final result:
\begin{multline*}
\left<\Delta \delta{F}_{ij}\right> = \frac{1}{T}\sum_k \frac{\hat M^{-1}_{ik}}{ \left<y^2_k\right>}\left[ \sum_l \delta F_{kl}C_{lj}-\sum_{lpr}\delta F_{kl}F_{rp}C_{lp}F_{rj} \right.\\
\left.- \sum_{lpr} F_{kl}\delta F_{rp}C_{lp}F_{rj}\right]-\frac 1T \delta F_{ij}+\mathcal{O}\left(\frac 1{T^2}\right). 
\end{multline*}
\end{enumerate}
\end{proof}
 
 \subsection{Proof of Theorem \ref{Th3}}\label{A5}

For ease of reference, we remind that in general $\delta{\bf F}$ can be written as,
\begin{align*}\tag{\ref{dF}}
\delta{\bf F} = \delta {\bf A} \,{\bf F} + \delta {\bf S}\, {\bf F}+\delta {\bf B}\, {\bf G}.
\end{align*}
Here, $\delta{\bf A}$ is an $m\times m$ skew symmetric matrix, $\delta{\bf S}$ is an $m\times m$ symmetric matrix and $\delta{\bf B}$ is an $m\times(n-m)$ matrix. ${\bf G}$ is an $(n-m) \times n$ matrix with orthonormal rows. These rows are chosen to be orthogonal to the rows of ${\bf F}$. Let ${\bf v}^{1,\ldots,n}$ be the eigenvectors ${\bf C}$ and $v^{1,\ldots,n}$ be the corresponding eigenvalues. We label them such that  ${\bf F}$ spans the same space as the space spanned by the first $m$ eigenvectors. We choose rows of ${\bf G}$ to be the remaining eigenvectors, i.e. ${\bf G}^\top:=[{\bf v}^{m+1},\ldots,{\bf v}^{n}]$. Then, for future reference,
\begin{align}\label{G}
{\bf F}{\bf G}^\top = 0, \quad {\bf G}{\bf G}^\top = {\bf I}_{(n-m)},\quad\text{and} \quad  \sum_k C_{ik}G^\top_{kj} = \sum_k C_{ik}v^{j+m}_k = v^{j+m} G^\top_{ij}.
\end{align}

We also remind the definition:
\begin{align*}\tag{\ref{Mhat}}
{\bf \hat M} := {\bf I}_m + {\bf M}.
\end{align*}

\begin{proof}[Proof of Theorem \ref{Th3}]
Below, we discuss the conditions under which perturbations of ${\bf F}$ are stable. We work to linear order in $T^{-1}$ as stated in Theorem \ref{Th3}. We treat separately the evolution of $\delta {\bf A}$, $\delta {\bf S}$ and $\delta {\bf B}$ under a general perturbation $\delta {\bf F}$ . 

\begin{enumerate}[leftmargin=*,label=\arabic*.]
\item{Stability of $\delta {\bf B}$} 

\begin{enumerate}[leftmargin=*,label*=\arabic*]
\item Evolution of $\delta {\bf B}$ is given by:
\begin{align}\label{BEvol}
 \left<\Delta \delta{B}_{ij}\right> = \frac 1T \sum_k \left( \frac{\hat M^{-1}_{ik}}{\left<y^2_k\right>} v^{j+m} -\delta_{ik}\right) \delta B_{kj}.
\end{align}
\begin{proof} Starting from \eqref{dF} and using \eqref{G}: 
\begin{align*}
 \left<\Delta \delta{B}_{ij}\right> &= \sum_{k}\left<\Delta \delta{ F}_{ik}\right> G^\top_{kj} \\
&= \frac{1}{T}\sum_k \frac{\hat M^{-1}_{ik}}{ \left<y^2_k\right>}  \sum_{lp} \delta F_{kl}C_{lp}G_{jp}- \frac{1}T \delta B_{ij}. 
\end{align*}
Here the last line results from equation \eqref{G} applied to \eqref{ddf}. Let's look at the first term again using \eqref{G} and then \eqref{dF},
\begin{align*}
 \frac{1}{T}\sum_k \frac{\hat M^{-1}_{ik}}{ \left<y^2_k\right>}  \sum_{lp} \delta F_{kl}C_{lp}G_{jp}&= \frac 1T  \sum_k \frac{\hat M^{-1}_{ik}}{\left<y^2_k\right>} \sum_l\delta F_{kl} v^{j+m} G_{jl}\nonumber \\
 &= \frac 1T  \sum_k \frac{\hat M^{-1}_{ik}}{\left<y^2_k\right>} v^{j+m} \delta B_{kj}.
\end{align*}
Combining these give \eqref{BEvol}.\end{proof}
 
\item When is \eqref{BEvol} stable? Next, we show that stability requires 
\begin{align*}
\left\lbrace v^{1}, \ldots,v^{m}\right\rbrace > \left\lbrace v^{m+1}, \ldots,v^{n}\right\rbrace.
\end{align*}

For ease of manipulation, we express \eqref{BEvol} as a matrix equation for each column of $\delta{\bf B}$. For convenience we change our notation to $\delta B_{kj}=\delta B_{k}^j$
\begin{align*}
 &\left<\Delta \delta{B}_{i}^j\right> = \sum_k P^j_{ik}\delta B_{k}^j\\
  &\text{where}\quad P^j_{ik}\equiv \frac{1}{T}\left(O_{ik}v^{j+m} -\delta_{ik}\right), \quad  \text{and} \quad O_{ik} \equiv \frac{\hat M^{-1}_{ik}}{\left<y^2_k\right>}. 
\end{align*}
We have one matrix equation for each $j$. These equations are stable if all eigenvalues of all ${\bf P^j}$ are negative. 
\begin{align*}
\lbrace \text{eig}({\bf P})\rbrace < 0 &\quad\implies \quad \lbrace \text{eig}({\bf O})\rbrace  < \frac{1}{v_j}, \quad j = m+1,\ldots,n.\\
&\quad\implies \quad \lbrace \text{eig}({\bf O}^{-1})\rbrace  > v_j, \quad j = m+1,\ldots,n.
\end{align*}

\item If one could calculate eigenvalues of ${\bf O}^{-1}$, the stability condition can be articulated. We start this calculation by noting that 
\begin{align}\label{Qinv1}
\sum_k O _{ik}\left<y_ky_j\right> &= \sum_{k}\hat M^{-1}_{ik}\frac{\left<y_ky_j\right>}{\left<y^2_k\right>} \nonumber \\
&= \sum_{k}\hat M^{-1}_{ik}\hat M_{kj} = \delta_{ij} & \text{(from \eqref{main2})}.
\end{align}
Therefore,
\begin{align}\label{Oeig}
{\bf O}^{-1} = \left<{\bf y} {\bf y}^\top\right> = {\bf F}{\bf C}{\bf F}^\top.
\end{align}
Then, we need to calculate the eigenvalues of ${\bf F}{\bf C}{\bf F}^\top$. They are:
\begin{align*}
 \text{eig}({\bf O}^{-1}) = \left\lbrace v^{1}, \ldots,v^{m}\right\rbrace.
 \end{align*}
\begin{proof} We start with the eigenvalue equation.
\begin{align*}
{\bf F}{\bf C}{\bf F}^\top{\boldsymbol \lambda} &= \lambda {\boldsymbol \lambda} 
\end{align*}
Multiply both sides by ${\bf F}^\top$:
\begin{align*}
 {\bf F}^\top{\bf F} {\bf C}{\bf F}^\top {\boldsymbol \lambda} = \lambda \left({\bf F}^\top{\boldsymbol \lambda}\right).
\end{align*}
Next, we use the commutation of  ${\bf F}^\top{\bf F}$ and ${\bf C}$, \eqref{FC}, and the orthogonality of neural filters, ${\bf F}{\bf F}^\top={\bf I}_m$, \eqref{orthF} to simplify the left hand side:
\begin{align*}
 {\bf F}^\top{\bf F} {\bf C}{\bf F}^\top {\boldsymbol \lambda} &=  {\bf C}{\bf F}^\top{\bf F} {\bf F}^\top {\boldsymbol \lambda} = {\bf C}\left({\bf F}^\top {\boldsymbol \lambda}\right).
\end{align*}
This implies that
\begin{align}\label{eq1}
 {\bf C}\left({\bf F}^\top {\boldsymbol \lambda}\right)=\lambda \left({\bf F}^\top{\boldsymbol \lambda}\right).
 \end{align}
Note that by orthogonality of neural filters, the following is also true:
\begin{align}\label{eq2}
 {\bf F}^\top{\bf F} \left({\bf F}^\top{\boldsymbol \lambda}\right)= \left({\bf F}^\top{\boldsymbol \lambda}\right).
 \end{align}
All the relations above would hold true if $\lambda = 0$ and $\left({\bf F}^\top{\boldsymbol \lambda}\right)=0$, but this would require ${\bf F}\left({\bf F}^\top{\boldsymbol \lambda}\right)={\boldsymbol \lambda}=0$, which is a contradiction. Then, \eqref{eq1} and \eqref{eq2} imply that $\left({\bf F}^\top{\boldsymbol \lambda}\right)$ is a shared eigenvector between ${\bf C}$ and ${\bf F}^\top{\bf F} $. ${\bf F}^\top{\bf F}$ and ${\bf C}$ was shown to commute before and they share a complete set of eigenvectors. However, some $n-m$ eigenvectors of  ${\bf C}$  have zero eigenvalues in ${\bf F}^\top{\bf F}$. We had labeled shared eigenvectors with unit eigenvalue in ${\bf F}^\top{\bf F}$ to be ${\bf v}^1,\ldots,{\bf v}^m$.  The eigenvalue of $\left({\bf F}^\top{\boldsymbol \lambda}\right)$  with respect to ${\bf F}^\top{\bf F} $ is  1, therefore ${\bf F}^\top{\boldsymbol \lambda} $ is one of ${\bf v}_1,\ldots,{\bf v}_m$. This proves that $\lambda = \left\lbrace v^{1}, \ldots,v^{m}\right\rbrace$ and 
\begin{align*}
 \text{eig}({\bf O}^{-1}) = \left\lbrace v^{1}, \ldots,v^{m}\right\rbrace.
 \end{align*}
\end{proof}
\item From \eqref{Oeig}, it follows that for stability 
\begin{align*}
\left\lbrace v^{1}, \ldots,v^{m}\right\rbrace > \left\lbrace v^{m+1}, \ldots,v^{n}\right\rbrace
\end{align*}
\end{enumerate}

\item{Stability of $\delta {\bf A}$ and $\delta {\bf S}$}

Next, we check stabilities of $\delta {\bf A}$ and $\delta {\bf S}$.
\begin{align}
 \left<\Delta \delta{A}_{ij}\right>+\left<\Delta \delta{S}_{ij}\right> &= \sum_{k}\left<\Delta \delta{ F}_{ik}\right> F^T_{kj} \qquad\qquad \text{(from definition \eqref{dF})} \nonumber \\
&=  -\frac{1}{T}\sum_k \frac{\hat M^{-1}_{ik}}{\left<y^2_k\right>} \sum_{lm} F_{kl}\delta F_{jm}C_{lm}-\frac 1T \left(\delta{A}_{ij}+\delta{S}_{ij}\right)\nonumber \\
&=  -\frac{1}{T}\sum_k \frac{\hat M^{-1}_{ik}}{ \left<y^2_k\right>} \sum_{l} \left({\bf F}{\bf C}{\bf F}^T\right)_{kl}\left(\delta{A}^T_{lj}+\delta{S}^T_{lj}\right) -\left(\delta{A}_{ij}+\delta{S}_{ij}\right) \label{dads}.
\end{align}
In deriving the last line, we used equations \eqref{dF} and \eqref{G}. The $k$ summation was calculated before \eqref{Qinv1}. Plugging this in \eqref{dads}, one gets
\begin{align*}
 \left<\Delta \delta{A}_{ij}\right>+\left<\Delta \delta{S}_{ij}\right> &= -\frac{1}{T}\left(\delta{A}_{ij}+\delta{A}^T_{ij}+\delta{S}_{ij}+\delta{S}_{ij}\right) = \frac{-2}{T}\delta{S}_{ij} \\
 \implies  \left<\Delta \delta{A}_{ij}\right> &= 0  \qquad\qquad \text{(from skew symmetry of {\bf A})}\\
\implies  \left<\Delta \delta{S}_{ij}\right> &= \frac{-2}{T}\delta{S}_{ij} .
 \end{align*}
$\delta {\bf A}$ perturbation, which rotates neural filters to other orthonormal basis within the principal subspace, does not decay. On the other hand, $\delta {\bf S}$ destroys orthonormality and these perturbations do decay, making the orthonormal solution stable. 
\end{enumerate}
Collectively, the results above prove Theorem \ref{Th3}.
\end{proof}

 \subsection{Perturbation of the stationary state due to data presentation}\label{A6}
 
 Our discussion of the linear stability of the stationary point assumed general perturbations. Perturbations that arise from data presentation, 
\begin{align}\label{eq3}
\delta {\bf F} = \Delta {\bf F},
\end{align}
form a restricted class of the most general case, and have special consequences. Focusing on this case, we show that data presentations do not rotate the basis for extracted subspace in the stationary state.

We calculate perturbations within the extracted subspace. Using \eqref{dF} and \eqref{G}
\begin{align}\label{dff}
 \delta{\bf A} + \delta {\bf S} &=  \delta {\bf F} \,{\bf F}^\top \nonumber \\
 &= \Delta {\bf F} \,{\bf F}^\top & \text{from \eqref{eq3}} \nonumber \\
&= \hat{\bf M}^{-1}\left(\Delta {\bf W} - \Delta \hat {\bf M}\, {\bf F}\right){\bf F}^\top & \text{expand \eqref{QF} to first order in $\Delta$} \nonumber \\
&=\hat{\bf M}^{-1}\left(\Delta {\bf W} \, {\bf F}^\top - \Delta \hat {\bf M}\right) & \text{from \eqref{orthF}} .
\end{align}
Let's look at $\Delta {\bf W} \, {\bf F}^\top$ term more closely:
\begin{align*}
\left(\Delta {\bf W} \, {\bf F}^\top \right)_{ij} &= \sum_k \eta_i\left(y_i x_k - y_i^2 W_{ik}\right)F^\top_{kj} \nonumber \\
&=\eta_i\left(y_i \sum_k F_{jk} x_k - y_i^2 \sum_kW_{ik}F^\top_{kj} \right) \nonumber \\
& =\eta_i\left(y_i y_k - y_i^2 \hat M_{ij} \right) \nonumber \\ 
&= \Delta \hat{M}_{ij}.
\end{align*}
Plugging this back into \eqref{dff} gives, 
\begin{align}
\delta{\bf A} + \delta {\bf S} = 0, \qquad \implies \qquad \delta{\bf A} =0, \quad \& \quad \delta {\bf S} = 0,
\end{align}
Therefore, perturbations that arise from data presentation do not rotate neural filter basis within the extracted subspace. This property should increase the stability of the neural filter basis within the extracted subspace.


{
\bibliographystyle{apa}

}
\end{document}